\documentclass[11pt]{article}

\usepackage{amsmath,amssymb,amsthm,amsfonts,latexsym,bbm,xspace,graphicx,float,mathtools}
\usepackage[backref,colorlinks,citecolor=blue,bookmarks=true]{hyperref}
\usepackage[letterpaper,margin=1in]{geometry}
\newcommand{\tlog}{{\rm logt}}

\usepackage{tikz}
\newcommand{\costasnote}[1]{{#1}}

\newcommand{\jdnote}[1]{{#1}}

\usepackage{tikz}

\usepackage{algorithm2e}

\newtheorem{Theorem}{Theorem}
\newtheorem{Theorem*}{Theorem}

\newtheorem{Claim}[Theorem]{Claim}
\newtheorem{Claim*}[Theorem]{Claim}
\newtheorem{Corollary}[Theorem]{Corollary}

\newtheorem{CounterExample*}{$\overline{\hbox{\bf Example}}$}
\newtheorem{Definition}[Theorem]{Definition}

\newtheorem{Example*}[Theorem]{Example}

\newtheorem{Intuition*}[Theorem]{Intuition}
\newtheorem{Joke*}[Theorem]{Joke}
\newtheorem{Lemma}[Theorem]{Lemma}
\newtheorem{Lemma*}[Theorem]{Lemma}
\newtheorem{Open problem}[Theorem]{Open problem}

\newtheorem{Question*}[Theorem]{Question}

\def \bSubexa    {\begin{subexa}}

\newcommand{\ignore}[1]{}

\newcommand{\EE}{\mathbb{E}}

\def \cI     {{\cal I}}

\def \cP     {{\cal P}}
\def \cQ     {{\cal Q}}

\newcommand{\Var}{{\rm Var}}

\def \upto  {{,}\ldots{,}}

\newcommand{\ed}{\stackrel{\mathrm{def}}{=}}

\def\ignore#1{}

\newcommand{\bi}{\begin{itemize}}
\newcommand{\ei}{\end{itemize}}

\def\orpro{\mathop{\mathchoice
   {\vee\kern-.49em\raise.7ex\hbox{$\cdot$}\kern.4em}
   {\vee\kern-.45em\raise.63ex\hbox{$\cdot$}\kern.2em}
   {\vee\kern-.4em\raise.3ex\hbox{$\cdot$}\kern.1em}
   {\vee\kern-.35em\raise2.2ex\hbox{$\cdot$}\kern.1em}}\limits}

\def\andpro{\mathop{\mathchoice
 {\wedge\kern-.46em\lower.69ex\hbox{$\cdot$}\kern.3em}
 {\wedge\kern-.46em\lower.58ex\hbox{$\cdot$}\kern.25em}
 {\wedge\kern-.38em\lower.5ex\hbox{$\cdot$}\kern.1em}
 {\wedge\kern-.3em\lower.5ex\hbox{$\cdot$}\kern.1em}}\limits}

\def\simge{\mathrel{%
   \rlap{\raise 0.511ex \hbox{$>$}}{\lower 0.511ex \hbox{$\sim$}}}}

\def\simle{\mathrel{
   \rlap{\raise 0.511ex \hbox{$<$}}{\lower 0.511ex \hbox{$\sim$}}}}

\newcommand{\Pbd}{PBD}
\newcommand{\pbd}[1]{{\cal PBD}_{#1}}
\newcommand{\pbdn}{\pbd \absz}

\newcommand{\pbdep}{\pbd{\absz,\epsilon}}
\newcommand{\pbdabsz}{\pbd{\absz}}
\newcommand{\ppbd}{\dP_{\rm pbd}}

\newcommand{\Binomial}[2]{Binomial({#1},{#2})}
\newcommand{\absz}{n}
\newcommand{\abszp}{\absz'}
\newcommand{\nsmp}{k}
\newcommand{\Nsmp}{K}
\newcommand{\Nsmpi}[1]{\Nsmp_{#1}}

\newcommand{\dst}{\epsilon}
\newcommand{\epsp}{\epsilon'}
\newcommand{\err}{\delta}

\newcommand{\dtv}[2]{d_{TV}(#1,#2)}
\newcommand{\ellone}[2]{d_{TV}(#1,#2)}

\newcommand{\elltwosq}[2]{\ell_2^2(#1,#2)}
\newcommand{\ellto}{\ell_2}

\newcommand{\prob}[1]{{\rm Prob}\left(#1\right)}
\newcommand{\Poisson}{{\rm poi}}
\newcommand{\poi}[1]{\Poisson\left(#1\right)}
\newcommand{\poid}[2]{\Poisson(#1,#2)}

\newcommand{\Pbin}{P_{b}}
\newcommand{\dP}{P}

\newcommand{\cPp}{\cP'}
\newcommand{\cQp}{\cQ'}

\newcommand{\Pone}{P_1}
\newcommand{\Ptwo}{P_2}

\newcommand{\tpone}{P_{tp1}}
\newcommand{\tptwo}{P_{tp2}}

\newcommand{\mltone}{\mlt_1}
\newcommand{\sigmaone}{\sigma_1^2}

\newcommand{\mlttwo}{\mlt_2}
\newcommand{\sigmatwo}{\sigma_2^2}

\newcommand{\Ptrue}{Q}
\newcommand{\Pund}{P}

\newcommand{\lmbi}{\lambda_i}
\newcommand{\lmbip}{\lambda_i'}

\newcommand{\lamb}[1]{\lambda_{#1}}
\newcommand{\lambp}[1]{\lambda_{#1}'}

\newcommand{\mlt}{\mu}
\newcommand{\var}{Var}
\newcommand{\std}{\sigma}

\newcommand{\bern}{B}
\newcommand{\low}{L}
\newcommand{\high}{H}
\newcommand{\intlh}{[L,H]}

\newcommand{\insmp}{S}

\newcommand{\intsize}{H-L}

\newcommand{\tO}{\tilde{O}}

\newcommand{\ypbd}{{\bf Yes \Pbd\xspace}}
\newcommand{\npbd}{{\bf No \Pbd}}
\newcommand{\mean}{\mu}
\newcommand{\meanp}{\mu'}
\newcommand{\varp}{\sigma'}

\newcommand{\quantile}[1]{N_{#1}}
\newcommand{\tp}{\hat{P}_{tp}}
\newcommand{\tpund}{P_{tp}}
\newcommand{\tpunder}{P_{tp}({\mlt},{\sigma}^2)}
\newcommand{\tpest}{P_{tp}(\hat{\mlt},\hat{\sigma}^2)}

\newcommand{\hvar}{\hat{\sigma}}
\newcommand{\hmlt}{\hat{\mlt}}
\newcommand{\XoN}{X_1^{\Nsmp}}

\title{Testing Poisson Binomial Distributions}
\author{Jayadev Acharya\thanks{Supported by grant from MITEI-Shell program.}\\
EECS, MIT\\
\tt{jayadev@csail.mit.edu}
\and
Constantinos Daskalakis\thanks{Supported by a Sloan Foundation Fellowship, a Microsoft Research Faculty Fellowship and NSF Award CCF-0953960 (CAREER) and CCF-1101491.}\\
EECS, MIT \\
\tt{costis@mit.edu}
}
\begin{document}
\addtocounter{page}{-1}

\maketitle \thispagestyle{empty}

\begin{abstract}
A Poisson Binomial distribution over $n$ variables is the distribution of the sum of $n$ independent Bernoullis. We provide a sample near-optimal algorithm for testing whether a distribution $\Pund$ supported on $\{0,\ldots,n\}$ to which we have sample access is a Poisson Binomial distribution, or far from all Poisson Binomial distributions. The sample complexity of our algorithm is $O(n^{1/4})$ to which we provide a matching lower bound. We note that our sample complexity improves quadratically upon that of the naive ``learn followed by tolerant-test'' approach, while instance optimal identity testing~\cite{ValiantV14} is not applicable since we are looking to simultaneously test against a whole family of distributions.
\end{abstract}

\newpage

\section{Introduction} \label{sec:intro}

Given independent samples from an unknown probability distribution $\Pund$ over $\{0,\ldots,n\}$, can you explain $\Pund$ as the distribution of the sum of $n$ independent Bernoullis? For example, $\Pund$ may be the number of faculty attending the weekly faculty meeting, and you may be looking to test whether your observations are consistent with the different faculty decisions being independent. It is a problem of testing against a {\em family} of distributions: 

\medskip \noindent {\sc PbdTesting:} Given $\dst>0$ and sample access to an unknown distribution $\Pund$ over $\{0,\ldots,n\}$, test whether $\Pund \in \pbdabsz$, or $\dtv{\Pund}{\pbdabsz}>\dst$, where $\pbdabsz$ is the set of Poisson Binomial distributions over $n$ variables.

\medskip Besides any practical applications, the theoretical interest
in studying {\sc PbdTesting}, and for that matter testing membership
to other classes of distributions, stems from the fact that ``being a
Poisson Binomial distribution'' is not a symmetric property of a
distribution; hence the results of~\cite{ValiantV11a} cannot be
brought to bear. At the same time, ``being a Poisson Binomial
Distribution''  does not fall into the shape restrictions to a
distribution, such as
uniformity~\cite{GoldreichRon00,BatuFFKRW01,Paninski08, AcharyaJOS13b} or
monotonicity~\cite{BatuKR04}, for which (near-)optimal testing
algorithms have been obtained. While there has been a lot of work
on learning distributions from a class of
distributions~\cite{ChanDSS13,FeldmanOS05,MoitraV10, BelkinS10}, there
is still a large gap in our current knowledge about the complexity of
testing against general 
families of distributions, \costasnote{unless both the unknown
  distribution and the family have been restricted a
  priori~\cite{DaskalakisDSVV13,DiakonikolasKN14}.}

\jdnote{More broadly, our problem falls in the general framework of property
  testing~\cite{Goldreich98, Fischer01, Rubinfeld06, Ron08}, where the
  objective is to decide if an object  
  (in our case a distribution) has a certain property, or is \emph{far} from having the
  property. The objective is to solve such problems while minimizing the 
  number of queries to the object (in our case samples from the distribution) computationally efficiently. In the context of property testing of distributions, our objective here is different from
  the classic problem of (composite) hypothesis testing in
  statistics~\cite{Fisher25, LehmannR05},
  where the number of samples are taken to infinity to study error exponents and consistency.}

\smallskip An obvious approach to {\sc PbdTesting} is to learn a candidate Poisson Binomial distribution~$\Ptrue$ that is $\dst/2$-close to $\Pund$, if $\Pund$ truly is a Poisson Binomial distribution. This is known to be quite cheap, only requiring  $\tilde{O}(1/\dst^2)$ samples from $\Pund$~\cite{DaskalakisDS12}. We can then use a tolerant tester to test $\dtv{\Pund}{\Ptrue} \le \dst/2$ vs $\dtv{\Pund}{\Ptrue}>\dst$. Such a tester would allow us to distinguish $\Pund \in \pbdabsz$ vs $\dtv{\Pund}{\pbdabsz}>\dst$, as $\dtv{\Pund}{\Ptrue} \le \dst/2 \Leftrightarrow \Pund \in \pbdabsz$. 

Given that any $\Ptrue \in \pbdabsz$ has effective support $O(\sqrt{n \log{1/\dst}})$,\footnote{{\em Effective support} is the smallest  set of contiguous integers where the distribution places all but $\epsilon$ of its probability mass.} we can easily construct a tolerant tester that uses $\tilde{O}(\sqrt{n}/\dst^2)$ samples, resulting in overall sampling complexity of $\tilde{O}(\sqrt{n}/\dst^2)$. On the other hand, we do not see how to substantially improve this approach, given the lower bound of $\Omega(m/\log m)$ for tolerant identity testing distributions of support size $m$~\cite{ValiantV11a}.

\smallskip A somewhat different approach would circumvent the use of a tolerant identity tester, by exploiting the small amount of tolerance accommodated by known (non-tolerant) identity testers. For instance, \cite{BatuFFKRW01} show that, given a distribution $\Ptrue$ of support $m$ and $\tilde{O}(\sqrt{m}) \cdot {\rm poly}(1/\dst)$ samples from an unknown distribution $\Pund$ over the same support, one can distinguish $\dtv{\Pund}{\Ptrue}\le\frac{\dst^3}{4\sqrt{m}\log m}$ vs $\dtv{\Pund}{\Ptrue}>\dst$. Hence, we can try to first find a candidate $\Ptrue\in\pbdabsz$ that is $\frac{\dst^3}{4\sqrt{m}\log m}$-close to $\Pund$, if $\Pund\in\pbdabsz$, and then do (non-tolerant) identity testing against $\Ptrue$. In doing so, we can use $m = O(\sqrt{n \log{1/\dst}})$, since that is the worst case effective support of $\Pund$, if $\Pund\in\pbdabsz$. 

The testing step of this approach is cheaper, namely $\tilde{O}(n^{1/4}) \cdot {\rm poly}(1/\epsilon)$ samples, but now the learning step becomes more expensive, namely $\tilde{\Omega}(\sqrt{n}) \cdot {\rm poly}(1/\dst)$ samples, as the required learning accuracy is more extravagant than before.

\smallskip Is there then a fundamental barrier, imposing a sample complexity of $\tilde{\Omega}(\sqrt{n})$? We show that the answer is ``no,'' namely

\begin{Theorem} \label{thm:main theorem}
For $n, \dst, \err >0$, there exists an algorithm, {\tt Testing PBDs}, that uses
 \[
O\left(\frac{\absz^{1/4} \sqrt{\log (1/\dst)}}{\dst^2}+\frac{\costasnote{\log^{2.5}
  (1/\dst)}}{\dst^6}\right) \cdot \log(1/\err)
\]
independent samples from an unknown distribution $\Pund$ over
$\{0,\ldots,n\}$ and, with probability $\ge 1 - \err$, outputs \ypbd,
if $\Pund\in\pbdn$, and \npbd, if $\dtv{\Pund}{\pbdn}>\dst$. The time
complexity of the algorithm is 
\[
O\Big({\absz^{1/4}\sqrt{\log (1/\dst)}}/{\dst^2}+ (1/\dst)^{O(\log^21/\dst)}\Big) \cdot \costasnote{\log{(1/ \delta)}}.
\]
\end{Theorem}
\noindent The proof of Theorem~\ref{thm:main theorem} can be found in Section~\ref{sec:proof of main result}. We also show that the dependence of our sample complexity on $n$
cannot be improved, by providing a matching lower bound in Section~\ref{app:lower_bound} as follows.
\begin{Theorem} \label{thm:lower bound}
Any algorithm for {\sc PbdTesting} requires $\Omega(n^{1/4}/\dst^2)$ samples. 
\end{Theorem}
\noindent One might be tempted to deduce Theorem~\ref{thm:lower bound} from the lower bound  for identity testing against ${\rm Binomial}(\absz,1/2)$, which has been shown to require $\Omega(\absz^{1/4}/\dst^2)$ samples~\cite{Paninski08,ValiantV14}. However, testing against a class of distributions may very well be easier than testing against a specific member of the class. (As a trivial example consider the class of all distributions over $\{0,\ldots,n\}$, which are trivial to test.) Still, for the class $\pbdabsz$, we establish the same lower bound as for ${\rm Binomial}(\absz,1/2)$, deducing that the dependence of our sample complexity on $\absz$ is tight up to constant factors, while the dependence on $\dst$ of the leading term in our sample complexity is tight~up~to~a~logarithmic~factor.

\subsection{Related work and our approach} \label{sec:approach}

Our testing problem is intimately related to the following fundamental problems:

%



\medskip \noindent {\sc IdentityTesting:} Given a known distribution $\Ptrue$ and independent samples from an unknown distribution $\Pund$, which are both supported on $[m]:=\{0,\ldots,m\}$, determine whether $\Pund=\Ptrue$ {\sc Or} $\dtv{\Pund}{\Ptrue}>\dst$. If $\dtv{\Pund}{\Ptrue} \in (0,\dst]$, then any answer is allowed.

\medskip \noindent {\sc Tolerant-IdentityTesting:} Given a known distribution $\Ptrue$ and independent samples from an unknown distribution $\Pund$, which are both supported on $[m]$, determine whether $\dtv{\Pund}{\Ptrue} \le \dst/2$ {\sc Or} $\dtv{\Pund}{\Ptrue}>\dst$. If $\dtv{\Pund}{\Ptrue} \in (\dst/2,\dst]$, then any answer is allowed.

\bigskip It is known that {\sc IdentityTesting} can be solved from a
near-optimal $\tilde{O}(\sqrt{m}/\dst^2)$ number of
samples~\cite{BatuFFKRW01,Paninski08}. The guarantee is obviously
probabilistic: with probability $\ge 2/3$, the algorithm outputs
``equal,'' if $\Pund = \Ptrue$, and ``different,'' if
$\dtv{\Pund}{\Ptrue}>\dst$. On the other hand, even testing whether $\Pund$ equals the uniform distribution over $[m]$ requires $\Omega(\sqrt{m}/\dst^2)$ samples.

While the identity tester of~\cite{BatuFFKRW01} allows in fact a
little bit of tolerance (namely distinguishing $\dtv{\Pund}{\Ptrue}\le
\frac{\dst^3}{4\sqrt{m}\log m}$ vs $\dtv{\Pund}{\Ptrue}>\dst$), it
does not accommodate a tolerance of $\dst/2$. Indeed, \cite{ValiantV11a} show that there is a gap in the sample complexity of tolerant vs non-tolerant testing, showing that the tolerant version requires $\Omega(m/\log m)$ samples. 

%

\smallskip As discussed earlier, these results on identity testing in conjunction with the algorithm of~\cite{DaskalakisDS12} for learning Poisson Binomial distributions can be readily used to solve {\sc PbdTesting}, albeit with suboptimal sample complexity. Moreover, recent work of Valiant and Valiant~\cite{ValiantV14} pins down the optimal sampling complexity of {\sc IdentityTesting} up to constant factors for any distribution $\Ptrue$, allowing sample-optimal testing on an instance to instance basis. However, their algorithm is not applicable to {\sc PbdTesting}, as it allows testing whether an unknown distribution $\Pund$ equals a specific distribution $\Ptrue$ vs being $\epsilon$-far from $\Ptrue$, but not testing whether $\Pund$ belongs to a class of distributions vs being $\epsilon$-far from all distributions in the class. 

\paragraph{Our Approach.} What we find quite interesting is that our ``learning followed by tolerant testing'' approach seems optimal: The learning algorithm of~\cite{DaskalakisDS12} is optimal up to logarithmic factors, and there is strong evidence that tolerant identity testing a Poisson Binomial distribution requires $\tilde \Omega(\sqrt{n})$ samples. So where are we losing?

We observe that, even though {\sc PbdTesting} can be reduced to tolerant identity testing of the unknown distribution $\Pund$ to a single Poisson Binomial distribution $\Ptrue$, we cannot consider the latter problem out of context, shooting at optimal testers for it. Instead, we are really trying to solve the following problem:

\medskip {\sc Tolerant-(Identity+PBD)-Testing:} Given a known $\Ptrue$ and independent samples from an unknown distribution $\Pund$, which are both supported on $[m]$, determine whether ($\dtv{\Pund}{\Ptrue} \le \dst_1$ {\em {\sc And} $\Pund \in \pbdabsz$}) {\sc Or} ($\dtv{\Pund}{\Ptrue}>\dst_2$). In all other cases, any answer is allowed.

\medskip The subtle difference between {\sc
  Tolerant-(Identity+PBD)-Testing} and {\sc Tolerant-IdentityTesting} is
the added clause ``{\sc And} $\Pund \in \pbdabsz$,'' which, it turns
out, makes a big difference for certain $\Ptrue$'s. In particular, we would hope that, when $\Ptrue$
and $\Pund$ are Poisson Binomial distributions with about the same variance, then the $\ell_1$
bound $\dtv{\Pund}{\Ptrue} \le \dst_1$ implies a good enough $\ell_2$
bound, so that {\sc Tolerant-(Identity+PBD)-Testing} can be reduced to
tolerant identity testing in the $\ell_2$ norm. We proceed to sketch
the steps of our tester in more detail.

\medskip The first step is to run the learning algorithm of~\cite{DaskalakisDS12} on $\tilde O(1/\epsilon^2)$ samples from $\Pund$. The result is some $\ppbd\in\pbdn$ such that, if $\Pund \in \pbdn$ then, with probability $\ge .99$, $\dtv{\Pund}{\ppbd}<\dst/10$. We then bifurcate depending on the variance of the learned $\ppbd$. For some constant $C$ to be decided, we consider the following cases.
\begin{itemize}
\item {\bf Case 1: $\std^2(\ppbd)< \frac{\costasnote{C \cdot \log^4{1/\epsilon}}}{\dst^8}$}

In this case, $\ppbd$ assigns probability mass of $\ge 1-\epsilon/5$ to an interval $\cI$ of size $\costasnote{O(\log^{2.5}(1/\dst)/\dst^4)}$. 
If $\Pund\in\pbdn$, then the
$\ell_1$ distance between $\Pund$ and $\ppbd$ over $\cI$ is at most
$\dst/5$ (with probability at least $0.99$). If $\dtv{\Pund}{\pbdn}>\dst$, then over the same interval, the
$\ell_1$ distance is at least $4\dst/5$. We can therefore do
tolerant identity testing restricted to support $\cI$, with $O(|\cI|/\dst^2)=\costasnote{O({\log^{2.5} (1/\dst)}/\dst^6)}$ samples. To this end, we use a simple tolerant identity test whose sample complexity is tight up to a logarithm in the support size, and very easy to analyze. Its use here does not affect the dependence of the overall sample complexity on $\absz$.  

\item {\bf Case 2: $\std^2(\ppbd) \ge \frac{\costasnote{C \cdot \log^4{1/\epsilon}}}{\dst^8}$}

We are looking to reduce this case to a  {\sc Tolerant-(Identity+PBD)-Testing} task for an appropriate distribution $Q$ that will make the reduction to tolerant identity testing in the $\ell_2$ norm feasible. First, it follows from~\cite{DaskalakisDS12} that in this case we can  actually
assume that $\ppbd$ is a Binomial distribution. 
Next, we use $O(n^{1/4}/\dst^2)$ samples to obtain estimates $\hmlt$ and $\hvar^2$ for the mean and variance of $\Pund$, and consider  the Translated Poisson
distribution $\tpest$ with parameters $\hmlt$ and $\hvar^2$; see Definition~\ref{def:translated poisson}.
If $\Pund \in \pbdabsz$, then with good probability (i) $\hmlt$ and $\hvar^2$ are extremely accurate as characterized by Lemma~\ref{lem:mean_var}, and (ii) using the Translated Poisson approximation to the Poisson Binomial distribution, Lemma~\ref{lem:pbd_tp}, we can argue that $\dtv{\Pund}{\tpest}\le  \costasnote{\frac{\dst^2}{10}}$. 

Before getting to the heart of our test, we perform one last check. We calculate an estimate of $\dtv{\tpest}{\ppbd}$ that is accurate to within $\pm \dst/5$. If our estimate $\hat{d}_{TV}(\Pbin,\ppbd) > \dst/2,$ we can safely deduce $\dtv{\Pund}{\pbdn} > \dst$. Indeed, if $\Pund \in \pbdn$, we would have seen $\hat{d}_{TV}(\tpest,\ppbd) \le \dtv{\tpest}{\ppbd} + \dst/5 \le \dtv{\tpest}{\Pund}+\dtv{\Pund}{\ppbd} + \dst/5 \le   \frac{\dst^2}{10} +  \frac{\dst}{10} +  \frac{\dst}{5} < 2\dst/5.$


If $\hat{d}_{TV}(\tpest,\ppbd) \le \dst/2,$ then
\[
\ellone{\Pund}{\tpest}\ge \ellone{\Pund}{\ppbd}-\ellone{\ppbd}{\tpest}
\ge \ellone{\Pund}{\ppbd}-\frac{7\dst}{10}.
\]
At this point, there are two possibilities we need to distinguish between:  {\em either} $\Pund \in \pbdn$ and $\dtv{\Pund}{\tpest}\le  \frac{\dst^2}{10}$, {\em or} 
$\dtv{\Pund}{\pbdn} > \dst$  and $\dtv{\Pund}{\tpest} \ge  \frac{3\dst}{10}$.\footnote{\costasnote{The two cases identified here correspond to Cases~\ref{good case 1} and~\ref{good case 2} of Section~\ref{sec:proof of main result}, except that we include some logarithmic factors in the total variation distance bound in Case~\ref{good case 1} for minute technical reasons.}} We argue that we can use an $\ellto$ test to solve this instance of {\sc Tolerant-(Identity+PBD)-Testing}. 

\end{itemize}

Clearly, we can boost the probability of error to any $\err>0$ by
repeating the above procedure $\log (1/\err)$ times and outputting the
majority. Our algorithm is provided as Algorithm~\ref{alg:test_pbd}.

\section{Preliminaries} \label{sec:preliminaries}

%

We provide some basic definitions, and state results that will be useful in our analysis.
\costasnote{\begin{Definition} The {\em truncated logarithm} function $\rm tlog$ is defined as $\tlog(x) = \max\{1, \log x\}$, for all $x \in (0,+\infty)$, where $\log x$ represents the natural logarithm of $x$.
\end{Definition}}
\begin{Definition}\label{def:effective support}
Let $P$ be a distribution over $[\absz]=\{0\upto\absz\}$. The $\epsilon$-{\em effective support of $P$} is the length of the smallest interval where the distribution places all but at most $\epsilon$ of its probability mass.
\end{Definition} 
\begin{Definition}
A {\em Poisson Binomial Distribution (PBD) over $[\absz]$}
is the distribution of $X=\sum_{i=1}^{\absz}X_i$, where the $X_i$'s are (mutually)
independent Bernoulli random variables. $\pbdabsz$ is the set of all 
Poisson Binomial distributions over $[\absz]$.
\end{Definition} 

\begin{Definition}
The {\em total variation distance} between two distributions $P$ and $Q$ over
a finite set $A$ is $\ellone{P}{Q}\ed {1 \over 2} \sum_{i \in A}{|P(i)-Q(i)|}$.
The total variation distance between two sets of distributions $\cP$ and
$\cQ$ is $\ellone{\cP}{\cQ}\ed \inf_{P\in\cP,Q\in\cQ}\ellone{P}{Q}$.
\end{Definition} 

%
We make use of the following learning algorithm, allowing us to learn an unknown Poisson Binomial distribution over $[n]$ from $\tilde O(1/\dst^2)$ samples (independent of $n$). Namely,

\begin{Theorem}[\cite{DaskalakisDS12}]
\label{thm:cds12}
For all $n, \dst>0$, there is an algorithm that uses $\tO(1/\dst^2)$ samples from an
unknown $\Pund\in\pbdabsz$ and outputs some $\ppbd\in\pbdn$ such that
\begin{itemize}
\item 
 ${\ppbd}$ is supported on an interval of length $O(1/\dst^3)$, 
\item
or ${\ppbd}$ is a Binomial distribution.
\end{itemize}
Moreover, with probability $\ge 0.99$, $\ellone{\ppbd}{\Pund}<{\dst}$ and, if the algorithm outputs a Binomial distribution, the standard deviation of the output distribution $\ppbd$ is within a factor of $2$ of the standard deviation of the unknown distribution $\Pund$. Furthermore, the running time of the algorithm is
 $\tilde O(\log n) \cdot (1/\dst)^{O(\log^21/\dst)}$.\qed
\end{Theorem}

In the same paper the authors show that the mean and variance of a Poisson Binomial distribution can be estimated using a few samples. They use  the empirical mean and variance estimates and bound the means
and variances of these estimates. 
\begin{Lemma}[\cite{DaskalakisDS12}]
\label{lem:mean_var}
For all $\dst'>0$, there is an algorithm that, using $O(1/\dst'^2)$ samples from an unknown \Pbd\ with mean $\mlt$ and variance $\sigma^2$, produces estimates $\meanp$, $\varp$
such that  
\[
|\mean-\meanp|<\dst'\cdot\sigma \ \ \ and \ \ \
|\sigma^2-\varp^2|<\dst'\cdot \sigma^2\sqrt{4+\frac1{\sigma^2}},
\]
with probability $\ge .99$.\qed
\end{Lemma}

Since Poisson Binomial variables are sums of indicators, the following
 bound is also helpful.
\begin{Lemma}[Chernoff bound for sums of Indicators~\cite{Levchenko}]
\label{lem:chernoff_indicator}
Let $X_1,\ldots,X_n$ be independent Bernoulli random variables, $X=X_1+\ldots+X_n$, and $\sigma^2=\Var(X)$. Then, for
$0<\lambda<2\sigma$,  
\[
\prob{|X-\EE[X]|>\lambda\sigma}<2e^{-\lambda^2/4}.
\]
\end{Lemma}

The following result is a tail bound on Poisson random variables
obtained from the Chernoff Bounds~\cite{MacherU05}.
\begin{Lemma}[\cite{AcharyaDJOPS12}]
\label{lem:tail_poi}
If $X$ is a Poisson $\lambda$ random variable, then for $x\ge\lambda$,
\[
\Pr(X \geq x) \leq \exp \left( -\frac{(x-\lambda)^2}{2x}\right),
\]
and for $x\le\lambda$,
\[
\Pr(X \leq x) \leq \exp \left( -\frac{(x-\lambda)^2}{2\lambda}\right).
\]
\end{Lemma}
Poisson Binomial distributions are specified by $n$ parameters, and consequently there has been
a great deal of interest in approximating them via distributions with only a
few parameters. One such class of distributions are \emph{Translated
  Poisson distributions}, defined next.
  
\begin{Definition}[\cite{Rollin07}] \label{def:translated poisson}
A {\em Translated Poisson distribution}, denoted $\tpund(\mlt,\sigma^2)$, is the distribution of a
random variable $Y=\lfloor\mlt-\sigma^2\rfloor + Z$, where $Z$ is a random variable distributed according to the Poisson distribution $\poi{\sigma^2+\{\mlt-\sigma^2\}}$. Here $\lfloor x\rfloor$ and  $\{x\}$ denote respectively the integral and fractional parts of $x$ respectively.
\end{Definition}

The following lemma bounds the total variation and $\ell_\infty$ distance
between a Poisson Binomial and a {Translated Poisson} distribution with the same mean and variance. The bound on total variation distance is taken directly from~\cite{Rollin07}, while the $\ell_\infty$ bound is obtained via simple substitutions in their  $\ell_\infty$ bound.  
\begin{Lemma}[\cite{Rollin07}]
\label{lem:pbd_tp}
Let $X=\sum_i X_i$, where the $X_i$'s are independent Bernoulli random variables, $X_i\sim \bern(p_i)$. Also, let $q_{\max}=\max_k
\Pr(X=k)$, $\mlt=\sum p_i$ and $\sigma^2=\sum p_i(1-p_i)$. The following hold: 
\begin{align*}
\ellone{X}{\tpunder}&\le
\frac{2+\sqrt{\sum p_i^3(1-p_i)}}{\sum  p_i(1-p_i)};\\ 
\ell_{\infty}\left(X,  \tpund(\mlt,\sigma^2)\right)&\le \frac{2+2\sqrt{q_{\max}\sum
    p_i^3(1-p_i)}}{\sum  p_i(1-p_i)};\\
q_{\max}&\le \ellone{X}{ \tpund(\mlt,\sigma^2)}+ \frac1{2.3\sigma}.
\end{align*}
\end{Lemma}

Finally, the total variation distance between two Translated Poisson distributions can be bounded as follows. 
\begin{Lemma}[\cite{BarbourL07}]
\label{lem:barbour}
Let $\tpone$ and $\tptwo$ be Translated Poisson distributions with
parameters $(\mltone$, $\sigmaone)$ and $(\mlttwo$, $\sigmatwo)$
respectively. Then, 
\[
\ellone{\tpone}{\tptwo}\le \frac{|\mltone-\mlttwo|}{\min\{\sigma_1,\sigma_2\}}+\frac{|\sigmaone-\sigmatwo|+1}{\min\{\sigmaone,\sigmatwo\}}.
\]
\end{Lemma}


\section{Testing PBD's} \label{sec:proof of main result}
We fill in the details of the outline provided in
Section~\ref{sec:approach}. Our algorithm is given in 
Algorithm~\ref{alg:test_pbd}.
\RestyleAlgo{boxruled}

\begin{algorithm}[ht]
\label{alg:test_pbd}
\caption{\texttt{Testing PBDs}}
\KwIn {Independent samples from unknown $P$ over $[\absz]$, $\dst$, $\err>0$}
\KwOut{With probability \costasnote{$\ge 0.75$}, output \ypbd, if $P\in\pbdn$, and \npbd, if
  $\ellone{P}{\pbdn}>\epsilon$. }
Using \costasnote{$\tilde{O}(1/\dst^2)$} samples, run the algorithm of
\cite{DaskalakisDS12} with accuracy $\dst/10$ to obtain $\ppbd \in \pbdn$\;
\eIf{
$\std(\ppbd)^2<{{\costasnote{C \cdot \tlog^4{1/\epsilon}}} \over \dst^8}$ 
}{
Find an interval $\cI$ of length $\costasnote{O}({\tlog^{2.5}(1/\dst)}/\dst^4)$ such that $\ppbd(\cI) \ge 1- \epsilon/5$\;
Run {\tt Simple Tolerant Identity Test} to compare $P$ and $\ppbd$ on $\cI$, using $\costasnote{k}=O\left(\frac{\costasnote{{\tlog^{2.5} (1/\dst)}}}{\dst^{6}}\right)$ samples\;
\eIf {\texttt{Simple Tolerant Identity test} outputs {\bf $close$}}
{
output \ypbd\;
}
{
output \npbd\;
}
}{
Use $O(n^{1/4}/\epsilon^2)$ samples to estimate the mean, $\hmlt$, and variance, $\hvar^2$, of $P$\;
Calculate an estimate $\hat{d}_{TV}(\tpest,\ppbd)$ of $\ellone{ \tpest}{\ppbd}$ that is accurate to within $\pm \dst/5$\;
\If{
$\hat{d}_{TV}(\tpest,\ppbd)>\dst/2$ OR $\hvar^2>n/2$
}{
output \npbd\;
}
Draw $K \sim \poi{k}$ samples, where $k \ge C_1\frac{\sqrt{\hvar \cdot \tlog(1/\dst)}}{\dst^2}$ and $C_1$ is as determined by Lemma~\ref{lem:largel2}\;
Let $\Nsmp_i$ be the number of samples that equal $i$\;
{
\eIf{$\frac1{\nsmp^2}\sum\Big[(\Nsmp_i-\nsmp \tpest(i))^2-\Nsmp_i\Big]<
  \costasnote{0.25 \cdot {{c \dst^2 }\over {\hvar\sqrt{\tlog(1/\dst)}}}}$, where $c$ is the constant from Claim~\ref{lem:l2far},
}{output \ypbd}{output \npbd}
}}
\end{algorithm}


 Our algorithm starts by running the algorithm of Theorem~\ref{thm:cds12} with accuracy $\dst/10$ to find $\ppbd\in\pbdabsz$. If the unknown 
distribution $\Pund \in \pbdn$, then with probability $\ge 0.99$,
$\ellone{\Pund}{\ppbd}\le\dst/10$. 

As in the outline, we next consider two cases, depending on the variance of $\ppbd$:\footnote{Notice that in defining our two cases we use the truncated logarithm  instead of the logarithm function in our threshold variance. This choice is made for trivial technical reasons. Namely, this logarithmic factor will appear in denominators later on, and it is useful to truncate it to avoid singularities.}
\begin{itemize}
\item {\bf Sparse Case:} when $Var(\ppbd)<\frac{\costasnote{C \cdot \tlog^4{1/\epsilon}}}{\dst^8}$. 
\item{\bf Heavy case:} when $Var(\ppbd)\ge\frac{\costasnote{C \cdot \tlog^4{1/\epsilon}}}{\dst^8}$. 
\end{itemize}

Clearly, if the distribution $\ppbd$ given by Theorem~\ref{thm:cds12} is supported on an interval of length $O(1/\dst^3)$, then we must be in the sparse case. Hence, the only way we can be in the heavy case is when $\ppbd$ is a Binomial distribution with variance larger than our threshold. We treat the two cases separately next. 



\subsection{Sparse case} \label{sec:sparse}

Our goal is to perform a simple tolerant identity test to decide whether $\dtv{P}{\ppbd} \le \epsilon/10$ or $\dtv{P}{\ppbd} > \epsilon$. We first develop the tolerant identity test. 

\paragraph{Simple Tolerant Identity Test:} The test is describe in
Algorithm~\ref{alg:tolerant_identity} and is based on the 
 folklore result described as Lemma~\ref{lem:empirical}.
\RestyleAlgo{boxruled}
\begin{algorithm}[ht]
\label{alg:tolerant_identity}
\caption{\texttt{ Simple Tolerant Identity Test}}
\KwIn {Known $\Ptrue$ over finite set $A$ of cardinality $|A|=m$, $\dst >0$, and independent samples $X_1,\upto X_{k}$
from unknown $\Pund$}
\KwOut{Close, if $\ellone{P}{Q}\le\epsilon/10$, and far, if $\ellone{P}{Q}>2\epsilon/5$}
$\hat{P}\gets$ empirical distribution of $X_1^{k}$\;
\eIf{$\ellone{\hat{P}}{Q}<2.5\dst/10$
}{output close\;}{output far\;}
\end{algorithm}

\begin{Lemma}
\label{lem:empirical}
Let $\epsilon>0$, and $P$ be an arbitrary distribution over a finite set $A$ of size $|A|=m$. With $O(m/\dst^2)$
independent samples from $P$, we can compute a distribution $Q$ over $A$ such that $\dtv{P}{Q} \le \epsilon$, with probability at least $.99$. In fact, the empirical distribution achieves this bound. 
\end{Lemma}

Lemma~\ref{lem:empirical} enables the simple tolerant identity tester,
whose pseudocode is given in Algorithm~\ref{alg:tolerant_identity}, which takes $O(m/\dst^2)$ samples from a distribution $P$ over $m$ elements and outputs whether it is $\le \dst/10$ close or $>2\dst/5$ far from 
a known distribution $Q$. The simple idea is that with sufficiently
many samples, the empirical distribution $\hat{P}$ satisfies
$\dtv{P}{\hat{P}} < \epsilon/10$ (by Lemma~\ref{lem:empirical}), which
allows us to distinguish between $\ellone{P}{Q}\le\epsilon/10$ and $\ellone{P}{Q}>2\epsilon/5$.

\paragraph{Finishing the Sparse Case:} Lemma~\ref{lem:chernoff_indicator} implies that there exists an interval $\cI$ of length $\costasnote{O({1 \over \epsilon^4}\cdot \tlog^{2.5} {1\over \epsilon})}$ such that $\ppbd(\cI) \ge 1- \epsilon/5$. Let us find such an interval $\cI$, and consider the distribution $P'$ that equals $P$ on $\cI$ and places all remaining probability on $-1$. Similarly, let us define $\ppbd'$ from $\ppbd$. It is easy to check that:
\begin{itemize}
\item if $P \in \pbdn$, then $\dtv{P'}{\ppbd'}\le\epsilon/10$, since $\dtv{P}{\ppbd}\le\epsilon/10$ and $P'$, $\ppbd'$ are coarsenings of $P$ and $\ppbd$ respectively.
\item if $\dtv{P}{\pbdn} > \epsilon$, then $\dtv{P'}{\ppbd'} > 2 \epsilon/5$. (This follows easily from the fact that $\ppbd$ places less than $\epsilon/5$ mass outside of $\cI$.)
\end{itemize}

Hence, we can use our simple tolerant identity tester (Algorithm~\ref{alg:tolerant_identity}) to distinguish between these cases from $O(|\cI|/\epsilon^2) = \costasnote{O({1 \over \epsilon^6} \cdot \tlog^{2.5} {1\over \epsilon})}$ samples.

\subsection{Heavy case}

In this case, it must be that $\ppbd={\rm Binomial}({\abszp},{p}$) and 
$\abszp p(1-p)\ge{{\costasnote{C \cdot \tlog^4{1 \over \epsilon}}} \over {\dst^{8}}}$. The high level plan for this case is the following:
\begin{enumerate}
\item First, using Theorem~\ref{thm:cds12}, we argue that, if $P \in \pbdn$, then its variance is also $\Omega\left({{\costasnote{\tlog^4{1 \over \epsilon}}} \over {\dst^{8}}}\right)$ large.
\item Next, we apply Lemma~\ref{lem:mean_var} with $O(n^{1/4}/\epsilon^2)$ samples to get estimates $\hat{\mu}$ and $\hat{\sigma}^2$ of the mean and variance of $P$. If $P \in \pbdn$, then these estimates are very accurate, with probability at least $0.99$, and, by Lemmas~\ref{lem:pbd_tp} and~\ref{lem:barbour}, the corresponding Translated Poisson distribution $\tpest$ is \costasnote{${\epsilon^2 \over C' \sqrt{\tlog {1 \over \epsilon}}}$}-close to $P$, for our choice of $C'$ (that we can tune by choosing $C$ large enough).
\item Then, with a little preprocessing, we can get to a state where we need to distinguish between the following, \costasnote{for any $C' \ge 10$ of our choice}:
\begin{enumerate}
\item
$\Pund\in\pbdabsz$ and $P$ is \costasnote{${\epsilon^2 \over C' \sqrt{\tlog {1 \over \epsilon}}}$}-\emph{close to} $\tpest$ OR \label{good case 1}
\item
$\dtv{P}{\pbdn} >\epsilon$ and $P$ is $3 \epsilon/10$-\emph{far from} $\tpest$. \label{good case 2}
\end{enumerate}
\item Finally, using the $\ell_\infty$ bound in Lemma~\ref{lem:pbd_tp}, we show that, if the first case holds, then $\Pund$ is close to $\tpund(\hat{\mlt},
\hat{\sigma}^2)$ even in the $\ell_2$ distance. Using this, we show that it suffices to design
an $\ell_2$ test against $\tpest$. The computations of this algorithm
are similar to the $\chi-$squared statistic used for testing closeness
of distributions in~\cite{AcharyaDJOPS12, ChanDVV13}.
\end{enumerate}

We proceed to flesh out these steps:

\paragraph{Step 1:} By Theorem~\ref{thm:cds12}, if $\Pund\in\pbdabsz$ and a Binomial$(n',p)$ is output by the algorithm of~\cite{DaskalakisDS12}
then $P$'s variance is within a factor of $4$ from $\abszp p(1-p)$, with probability at least $0.99$. Hence, if $\Pund\in\pbdabsz$ and we are in the heavy case, then we know that, with probability $\ge0.99$:
\begin{align}
Var(\Pund) >{{\costasnote{C \cdot \tlog^4{1 \over \epsilon}}} \over {4\dst^{8}}}. \label{eq:large variance}
\end{align}
Going forward, if $\Pund\in\pbdn$, we condition on~\eqref{eq:large variance}, which happens with good probability.

\paragraph{Step 2:} Let us denote by $\mu$ and $\sigma^2$ the mean and variance of the unknown $P$. If $P \in \pbdn$, then clearly $\sigma^2 \le n/4$. So let us use $\epsp=\frac{\dst}{(n/4)^{1/8}}$ in
Lemma~\ref{lem:mean_var} to compute estimates $\hat{\mu}$ and $\hat{\sigma}^2$ of $\mu$ and $\sigma^2$ respectively. Given that $\sigma>1$ for a choice of $C\ge 4$ in~\eqref{eq:large variance}, we get:
\begin{Claim}
\label{lem:mean_var_spcl}
If $\Pund\in\pbdn$ and $C\ge 4$, then the outputs $\hat{\mu}$ and $\hat{\sigma}^2$ of the algorithm of Lemma~\ref{lem:mean_var} computed from $O(n^{1/4}/\dst^2)$ samples from $P$ satisfy the following with probability $\ge 0.99$:
\begin{align}
|\mu-\hat{\mu}|<\frac{\dst}{\sigma^{1/4}}\sigma\ \ \ and \ \ \ |\sigma^2-\hat{\sigma}^2|<\costasnote{3}\frac{\dst}{\sigma^{1/4}}\sigma^2. \label{eq:very accurate estimates}
\end{align}
\end{Claim}

Using these bounds, we show next that $\tpest$ is a good approximator of $P$, if $P \in \pbdn$. 
%
\begin{Claim}
\label{lem:p_tp}
If $\Pund\in\pbdabsz$ and \eqref{eq:large variance},~\eqref{eq:very accurate estimates} hold, then \costasnote{for any constant $C'$} there exists large enough $C$:
\begin{align}
\ellone{\Pund}{\tpest}<\frac{\costasnote{3}}{\sigma}+\frac{\costasnote{14}\dst}{\sigma^{1/4}}\le \costasnote{\frac{\dst^2}{C' \sqrt{\tlog{1 \over \epsilon}}}} \label{eq:good bound 1}.
\end{align}
\end{Claim}
\begin{proof}
\costasnote{We first note  that for large enough $C$ we have $\sigma>256$ so~\eqref{eq:very accurate estimates} implies that ${\sqrt{7} \over 2}\sigma >\hat{\sigma}>\hat{\sigma}>{\sigma \over 2}$.} By the first bound of Lemma~\ref{lem:pbd_tp} we have that: 
\begin{align}
\label{eqn:tv_pbd}
\ellone{\Pund}{\tpunder}<\costasnote{\frac{2+\sigma}{\sigma^2}}<\frac{\costasnote{3}}{\sigma}.
\end{align}
Using~\eqref{eq:very accurate estimates} \costasnote{and $\hat{\sigma}>\sigma/2$}
in Lemma~\ref{lem:barbour} gives
\[
\ellone{\tpunder}{\tpest}<\costasnote{14}\frac{\dst}{\sigma^{1/4}}.
\]
So, from triangle inequality, $\ellone{\Pund}{\tpest}<\frac{\costasnote{3}}{\sigma}+\frac{\costasnote{14}\dst}{\sigma^{1/4}}$. Plugging~\eqref{eq:large variance} into this bound gives $\ellone{\Pund}{\tpest}\le \costasnote{\frac{\dst^2}{C' \sqrt{\tlog{1 \over \epsilon}}}}$, when $C$ is large enough. \end{proof}
Going forward, for any $C'$ of our choice (to be determined), we choose $C$ to be large enough as required by Claims~\ref{lem:mean_var_spcl} and~\ref{lem:p_tp}. In particular, this choice ensures $\sigma>256$, and ${\sqrt{7} \over 2}\sigma >\hat{\sigma}>{\sigma \over 2}$. Moreover, if $\Pund\in\pbdn$, we condition on~\eqref{eq:very accurate estimates} and~\eqref{eq:good bound 1}, which hold with good probability.

\paragraph{Step 3:} We do some pre-processing that allows us to reduce our problem to distinguishing between cases \ref{good case 1} and \ref{good case 2}. Given our work in Steps 1 and 2, if $\Pund\in\pbdabsz$, then with good probability, 
\begin{align}\ellone{\ppbd}{\tpest}\le \frac\dst{10}+\costasnote{\frac{\dst^2}{C' \sqrt{\tlog{1 \over \epsilon}}}}<\dst/5, \label{eq:useufl ineq 3}
\end{align}
for $C' \ge 10$. Given that $\tpest$ and $\ppbd$ are explicit distributions we can compute (without any samples) an estimate $\hat{d}_{TV}(\tpest,\ppbd)={d}_{TV}(\tpest,\ppbd) \pm \epsilon/5$. Based on this estimate we distinguish the following cases:
\begin{itemize}
\item If $\hat{d}_{TV}(\tpest,\ppbd) > \dst/2,$ we can safely deduce $\dtv{\Pund}{\pbdn} > \dst$. Indeed, if $\Pund \in \pbdn$, then by~\eqref{eq:useufl ineq 3} we would have  seen
\begin{align*}
&\hat{d}_{TV}(\tpest,\ppbd) \\\le& \dtv{\tpest}{\ppbd} + \dst/5 \le 2 \dst/5.
\end{align*}
\item If $\hat{\sigma}^2>n/2$, we can also safely deduce that $\dtv{\Pund}{\pbdn} > \dst$. Indeed, if $\Pund \in \pbdn$, then $\sigma^2 \le n/4$, hence $\hat{\sigma}^2 \le n/2$ by our assumption that ${\sqrt{7} \over 2}\sigma >\hat{\sigma}$.
\item So it remains to consider the case $\hat{d}_{TV}(\tpest,\ppbd) \le \dst/2.$ This implies
\begin{align}
&\ellone{\Pund}{\tpest}\nonumber\\
\ge &\ellone{\Pund}{\ppbd}-\ellone{\ppbd}{\tpest}\nonumber\\
\ge &\ellone{\Pund}{\ppbd}-\frac{7\dst}{10}. \label{eq: useful inequality 2}
\end{align}
Now, if $\dtv{P}{\pbdn} >\epsilon$, then \eqref{eq: useful inequality 2} implies $\ellone{\Pund}{\tpest} > {3 \epsilon \over 10}$. On the other hand, if $P \in \pbdn$, then~\eqref{eq:good bound 1} implies $\ellone{\Pund}{\tpest} \le \costasnote{\frac{\dst^2}{C' \sqrt{\tlog{1 \over \epsilon}}}}$, for any $C' \ge 10$ of our choice. So, it suffices to be able to distinguish between cases \ref{good case 1} and \ref{good case 2}.
\end{itemize}

%

\paragraph{Step 4:} We now show that an $\ell_2$ test suffices for distinguishing between Cases \ref{good case 1} and \ref{good case 2}. We start by bounding the $\ell_2$ distance between $P$ and $\tpest$ in the two cases of interest.  We start with the easy bound, corresponding to Case \ref{good case 2}.

\smallskip \noindent {\bf Case \ref{good case 2}:}
Using Cauchy-Schwarz Inequality, we can lower bound the $\ell_2$ distance
between $\Pund$ and $\tpest$ in this case.
\begin{Claim}
\label{lem:l2far}
$\ellone{\Pund}{\tpest}>3\dst/10$ implies:
\[
\elltwosq{\Pund}{\tpest}>\frac{c\dst^2}{\hat{\sigma}\sqrt{\tlog(1/\dst)}},
\]
for some absolute constant $c$.
\end{Claim}
\begin{proof}
By Lemma~\ref{lem:tail_poi},  $\tpest$ assigns
\costasnote{$\ge 1-\dst/10$} of its mass to an interval $\cI$ of length $\costasnote{O}(\hat{\sigma}\sqrt{\log
(1/\dst)})$. 
 Therefore, the $\ell_1$ distance between $\Pund$ and
$\tpest$ over $\cI$ is at least $\ellone{\Pund}{\tpest}-\costasnote{\dst/10} > 3\dst/10-\costasnote{\dst/10}>0.2\dst$. 
Applying the Cauchy-Schwarz Inequality, \emph{over this interval}: 
\[
\elltwosq{\Pund}{\costasnote{\tpest}}_{\cal I}\ge \frac{\ell_1^2(\Pund,\costasnote{\tpest})_{\cal I}}{|{\cal I}|}\ge
\frac{c\dst^2}{\hat{\sigma}\sqrt{\log (1/\dst)}}, 
\]
for some constant $c>0$. In the above inequality we denote by $\ell_1(\Pund,\tpest)_{\cal I}$ and $\ell_2(\Pund,\tpest)_{\cal I}$ the $\ell_1$ and $\ell_2$ norms respectively of the vectors obtained by listing the probabilities assigned by $P$ and $\tpest$ on all points in set ${\cal I}$.
\end{proof}

\smallskip \noindent {\bf Case \ref{good case 1}:}
Rollin's result stated in Section~\ref{sec:preliminaries} provides bounds on the 
$\ell_1$ and $\ell_\infty$ distance between a \Pbd\ and its
corresponding translated Poisson distribution. Using these bounds we can show:

\begin{Claim}
\label{lem:l2close}
In Case~\ref{good case 1}:
\[
\elltwosq{\Pund}{\tpest}\le \frac{\costasnote{5.3}}{\hat{\sigma}}\cdot \costasnote{\frac{2\dst^2}{C' \sqrt{\tlog{1 \over \epsilon}}}}.
\]
\end{Claim}
\begin{proof}
\costasnote{Recall that, by our choice of $C$, $\sigma>256$, hence~\eqref{eq:very accurate estimates} implies that ${\sqrt{7} \over 2}\sigma >\hat{\sigma}>{\sigma \over 2}>128$.} Next, we recall the following bound on the $\ell_2$ norm of any two distributions $P$ and $Q$:
\begin{align}
\label{eqn:l2linl1}
\elltwosq{\Pund}{Q}\le \ell_{\infty}({\Pund},{Q})\ell_1({\Pund},{Q}).
\end{align}
Claim~\ref{lem:p_tp} takes care of the $\ell_1$ term when we 
substitute $Q=\tpest$. We now bound the $\ell_\infty$ 
term. 
For any distributions
$P$ and $Q$ it trivially holds that $\ell_{\infty}(P,Q)<\costasnote{\max\{\max_i\{P(i)\},\max_i\{Q(i)\}\}}$. By the third part of Lemma~\ref{lem:pbd_tp} and Equation~\eqref{eqn:tv_pbd}:
\[
\max_i \Pund(i)\le \ellone{\Pund}{\tpunder}+\frac1{2.3\sigma}\le \frac{\costasnote{4}}{\sigma} \costasnote{\le {5.3 \over \hat{\sigma}}}.
\]
To bound $\max_i Q(i)$, a standard Stirling approximation on the definition of
Translated Poisson distribution shows that the
largest probability of $\tpest$ is at most \costasnote{$1.5/\hat{\sigma}$}. Hence, $\ell_{\infty}(P,\tpest) \le \costasnote{{5.3 \over \hat{\sigma}}}$.
Plugging the above bounds into Equation~\eqref{eqn:l2linl1} with $Q=\tpest$ shows the result. 
\end{proof}

\costasnote{Claims~\ref{lem:l2close}
and~\ref{lem:l2far} show that the ratio of the squared $\ell_2$
distance between $\Pund$ and $\tpest$ in Case~\ref{good case 2} versus Case~\ref{good case 1} can be made larger than any constant, by choosing $C'$ large enough. To distinguish between the two cases we employ a tolerant $\ell_2$ identity test, based on an unbiased estimator $T_n$ of the squared $\ell_2$ distance between $\Pund$ and $\tpest$ described in Section~\ref{sec:ell2 test}. We will see that distinguishing between Cases~\ref{good case 1} and~\ref{good case 2} boils down to showing that in the latter case:
\[
Var(T_n)\ll\EE[T_n]^2.
\]}

\subsubsection{Unbiased $\ell_2^2$ Estimator.} \label{sec:ell2 test}
Throughout this section we assume that distributions are sampled in the Poisson sampling framework, where the number of samples $K$ drawn from a distribution are distributed as a Poisson random variable of some mean $k$ of our choice, instead of being a fixed number $k$ of our choice. This simplifies the variance computations by inducing independence among the number of times each symbol appears, as we discuss next. \costasnote{Due to the sharp concentration of the Poisson distribution, the number of samples we draw satisfies $K\le 2k$, with probability at least $1-({e \over 4})^{k}$.}

Suppose $\Nsmp\sim\poi{\nsmp}$ samples $\XoN$ are generated from a
distribution $\Pone$ over 
$[\absz]$. Let $\Nsmpi i$ be the random variable denoting the number
of appearances of symbol 
$i$. Then $\Nsmpi i$ is distributed according to $\poi{\lamb i}$, where  $\lamb
i\ed{\costasnote{k} \Pone(i)}$, independently of all other $\Nsmpi j$'s. Let also
$\lambp i \ed \costasnote{k} \Ptwo(i)$, and define:

\begin{align}
T_n=T_{\absz}(\XoN,\Ptwo)=\frac1{\nsmp^2}\sum_{i\in[\absz]}\Big[(\Nsmpi i-\lambp i)^2-\Nsmpi i\Big]. \label{eq:costas2}
\end{align}
A straightforward, albeit somewhat tedious, computation involving Poisson moments shows that 
\begin{Lemma}
\label{lem:expvar}
\[
\EE[T_\absz]
=\elltwosq{\Pone}{\Ptwo}=\frac1{k^2}\sum_{i=1}^{\absz}(\lamb i-\lambp i)^2
\]
and 
\costasnote{\[
\var(T_\absz) = \frac{2}{k^4}\sum_{i=1}^{\absz} \Big[\lamb i^2+2\lamb i(\lamb
i-\lambp i)^2\Big]. 
\]}
\end{Lemma}
We use Lemma~\ref{lem:expvar} with $P_1=\Pund$ and $P_2=\tpest$ to bound the variance of $T_n$ in terms of its squared expected value in Case~\ref{good case 2}, where $\ellone{\Pund}{\tpest}>0.3\dst$. 
\begin{Lemma}
\label{lem:largel2}
There is an absolute constant $C_1$ such that if 
\[
\nsmp\ge C_1 \frac{\sqrt{\hvar \cdot \tlog (1/\epsilon)}}{\epsilon^2},
\]
and $\ellone{P}{\tpest}>0.3\epsilon$, then
\[
\Var(T_n)<\frac1{20}\EE[T_n]^2.
\]
\end{Lemma}

\begin{proof}

From Lemma~\ref{lem:expvar}:
\begin{align}
\var(T_\absz) = \frac{2}{k^4}\sum_{i=1}^{\absz} \lamb i^2+\frac{4}{k^4}\sum_{i=1}^{\absz} \lamb i(\lamb i-\lambp i)^2. \label{eq:costas1}
\end{align}
We will show that each term on the right hand side is less than $\EE[T_n]^2/40$:
\begin{itemize}
\item The second term can be bounded by using the following:
\begin{align*}
\sum_{i}\lamb i(\lambp i-\lamb i)^2
\stackrel{(a)}{\le} &\Big[\sum_{i}\lamb i^2\Big]^{\frac12} \Big[\sum_i (\lambp i-\lamb i)^4\Big]^{\frac12}\\
\stackrel{(b)}{\le} &\Big[\sum_{i}\lamb i^2\Big]^{\frac12} \Big[\sum_i (\lambp i-\lamb i)^2\Big],
\end{align*}
where $(a)$ uses the Cauchy-Schwarz inequality and $(b)$ follows 
from the fact that, for positive reals $a_1, \ldots a_n$, $(\sum_i a_i )^2\ge \sum
a_i^2$.
Therefore, to bound the second term of the right hand side of~\eqref{eq:costas1} by $\EE[T_n]^2/40$ it suffices to show that
\[
4\Big[\sum_{i}\lamb i^2\Big]^{\frac12} \Big[\sum_i (\lambp i-\lamb
i)^2\Big]<\frac1{40}{\Big[\sum_i (\lambp i-\lamb i)^2\Big]^2},
\]
which holds if 
\begin{align}
\label{eqn:firstsum}
\Big[\sum_{i}\lamb i^2\Big] <\frac1{160^2}{\Big[\sum_i (\lambp i-\lamb i)^2\Big]^2}.
\end{align}
\item To bound the first term of the right hand side of~\eqref{eq:costas1} by $\EE[T_n]^2/40$ it suffices to show that
\begin{align}
\label{eqn:secsum}
\Big[\sum_{i}\lamb i^2\Big] <\frac1{80}{\Big[\sum_i (\lambp i-\lamb i)^2\Big]^2}.
\end{align}
\end{itemize}
Note that~\eqref{eqn:firstsum} is stronger
than~\eqref{eqn:secsum}. Therefore, we only
prove~\eqref{eqn:firstsum}.
Recall, from the proof of Claim~\ref{lem:l2close}, that $\max_i \tpest (i)\le \frac{\costasnote{1.5}}{\hvar}$. Using $\lambp i = \nsmp \tpest(i)$ and $\sum \lamb i=\nsmp$, 
\begin{align*}
\sum_{i}(\lambp i-\lamb i)^2>& \sum_i\lamb i^2-2 \sum_i\lamb i\lambp i  \\
\ge &
\sum_i\lamb i^2-\frac{\costasnote{3}\nsmp}{\hvar} \sum_i\lamb i\\
=&\sum_i\lamb i^2-\frac{\costasnote{3}\nsmp^2}{\hvar}, 
\end{align*}
and hence 
\[
\sum_{i}(\lambp i-\lamb i)^2+\frac{\costasnote{3}\nsmp^2}{\hvar}> \sum_i\lamb i^2.
\]
Let $y\ed \sum_{i}(\lambp i-\lamb i)^2$. 
It suffices to show that
\[
\frac1{160^2}y^2> y+\frac{\costasnote{3}\nsmp^2}{\hvar}, 
\]
which holds if the following conditions are satisfied: $y>2\cdot160^2$ and $y^2>\costasnote{6}\cdot 160^2
\frac{k^2}{\hvar}$.
By Claim~\ref{lem:l2far},  
\begin{align*}
y=\sum_{i}(\lambp i-\lamb i)^2>\frac{\nsmp^2c\dst^2}{\hvar\sqrt{\tlog(1/\dst)}},
\end{align*}
so the conditions hold as long as: 
\[
k \ge C_1\frac{\sqrt{\hvar \cdot \tlog(1/\dst)}}{\dst^2},
\]
and $C_1$ is a large enough constant.
\end{proof}

\subsubsection{Finishing the Heavy Case.}

Recall that the ratio of the squared $\ell_2$ distance between $\Pund$ and $\tpest$ in Case~\ref{good case 2} versus Case~\ref{good case 1} can be made larger than any constant, by choosing $C'$ large enough. This follows from  Claims~\ref{lem:l2far} and~\ref{lem:l2close}. Let us choose $C'$ so that this ratio is $ >100$. Now let us draw $K \sim \poi{k}$ samples from the unknown distribution~$\Pund$, where $k \ge C_1\frac{\sqrt{\hvar \cdot \tlog(1/\dst)}}{\dst^2}$ and $C_1$ is determined by Lemma~\ref{lem:largel2}, and compute $T_n$ using~\eqref{eq:costas2} with $\Pone=\Pund$ and $\Ptwo=\tpest$. 

By Lemma~\ref{lem:expvar}, $\EE[T_\absz]=\elltwosq{\Pund}{\tpest}$. Moreover:
\begin{itemize}
\item In Case~\ref{good case 1}, by Markov's Inequality, $T_n$ does not exceed $10$ times its expected value with probability at least $0.9$. 

\item In Case~\ref{good case 2}, from Chebychev's Inequality and Lemma~\ref{lem:largel2} it follows that
\[
\prob{|T_n-\EE[T_n]|>\frac{\EE[T_n]}{\costasnote{\sqrt{2}}}}<\frac1{10}.
\]
\end{itemize}
It follows that we can distinguish between the two cases with probability at least \costasnote{$0.9$} by appropriately thresholding $T_n$. One possible value for the threshold is one quarter of the bound of Claim~\ref{lem:l2far}. This is the threshold used in Algorithm~\ref{alg:test_pbd}. This concludes the proof of correctness of the heavy case algorithm.

\subsection{Correctness and Sample Complexity of Overall Algorithm.}
We have argued the correctness of our algorithm conditioning on various events. The overall probability of correctness is at least \costasnote{$0.99^2 \cdot 0.9 \ge 0.75$}. Indeed, one $0.99$ factor accounts for the success of Theorem~\ref{thm:cds12}, if $P \in \pbdn$. If the algorithm continues in the sparse case, the second $0.99$ factor accounts for the success of the {\sc Simple Tolerant Identity Test}, and we don't need to pay the factor of $0.9$. If the algorithm continues in the heavy case, the second $0.99$ factor accounts for the success of Lemma~\ref{lem:mean_var}, if $P \in \pbdn$, and the $0.9$ factor accounts for the success of the $\ell_2$ test. (In this analysis, we have assumed that we use fresh samples, each time our algorithm needs samples from $\Pund$.) Clearly, running the algorithm $O(\log(1/\delta))$ times and outputting the majority of answers drives the probability of error down to any desired $\delta$, at a cost of a factor of $O(\log(1/\delta))$ in the overall sample complexity.

Let us now bound the sample complexity of our algorithm. It is easy to see that the expected number of samples is as desired, namely:  $O\left(\frac{\absz^{1/4} \sqrt{\tlog (1/\dst)}}{\dst^2}+\frac{\costasnote{\tlog^{2.5} (1/\dst)}}{\dst^6}\right)$ times a factor of $O(\log 1/\delta)$ from repeating $O(\log 1/\delta)$ times. (It is easy to convert this to a worst-case bound on the sample complexity, by adding an extra check to our algorithm that aborts computation whenever $K \ge \Omega(k)$.)

\section{Lower Bound}
\label{app:lower_bound}
We now show that any algorithm for {\sc PbdTesting} requires
$\Omega(\absz^{1/4}/\dst^2)$ samples, in the spirit
of~\cite{Paninski04, AcharyaDJOPS12}.

Our lower bound will be based on constructing two classes of
distributions $\cPp$ and $\cQp$ such that 
\begin{itemize}
\item[(a)]
$\cPp$ consists of the single distribution $P_0\ed\Binomial{n}{1/2}$.
\item[(b)]
a uniformly chosen $Q$ from $\cQp$ satisfies
$\ellone{Q}{\pbdabsz}>\dst$ with probability $>0.99$. 
\item [(c)]
any algorithm that succeeds in distinguishing $P_0$ from a uniformly chosen
distribution from $\cQp$ with probability
$>0.6$ requires $\ge c\cdot\absz^{1/4}/\dst^2$ samples, for an
absolute constant $c>0$.
\end{itemize}

Suppose $k_{min}$ is the least number of samples required for {\sc PbdTesting} with
success probability $>2/3$. We show that if the conditions above are
satisfied then 
\[
k_{min}\ge c\cdot\absz^{1/4}/\dst^2.
\]
The argument is straight-forward as we can use the {\sc PbdTesting} algorithm with $k_{min}$ samples to distinguish
$P_0$ from a uniformly chosen $Q \in \cQp$, by just checking whether $Q \in \pbdn$ or $\ellone{Q}{\pbdn}>\epsilon$. The success probability of the algorithm is at least $2/3 \cdot 0.99 > 0.6$. Indeed, by (b) a uniformly chosen distribution from
$\cQp$ is at least $\dst$ away from $\pbdabsz$ with probability
$>0.99$,  and the {\sc PbdTesting} algorithm succeeds with probability $>2/3$ on those distributions. Along with $(c)$ this proves the lower bound on $k_{min}$.

We now construct $\cQp$. Since $P_0$ is $Binomial(n,1/2)$, 
\[
P_0(i)= {\absz \choose i}\left(\frac12\right)^\absz,
\]
and $P_0(i)=P_0(\absz-i)$. 

Without loss of generality assume $n$ is even. For each of the
$2^{n/2}$ vectors 
\costasnote{$z_0z_1\ldots z_{n/2-1}\in\{-1,1\}^{n/2}$}, define a distribution $Q$
over $\{0,1,\ldots,n\}$ as follows, where $c$ is an absolute constant
specified later.  
\[
Q(i)=\begin{cases}
(1-c\dst z_i)P_0(i)& if\ \ i<n/2,\\
\costasnote{P_0(n/2)}& if\ \ i=n/2,\\
(1+c\dst z_{n-i})P_0(i)& otherwise.
\end{cases}
\]
 The class $\cQp$ is the collection of these $2^{n/2}$ distributions. We proceed to prove $(b)$ and $(c)$. 

{\bf Proof of Item (b):} We need to prove that a uniformly picked distribution in
$\cQp$ is $\dst-$far from $\pbdabsz$ with probability $>0.99$. Since
Poisson Binomials are log-concave, and hence unimodal, it will suffice
to show that in fact distributions in $\cQp$ are $\dst-$far from all
unimodal distributions. The intuition for this is that when a
distribution is picked at random from $\cQp$ it is equally likely to
be above $P_0$ or under $P_0$ \costasnote{at any point $i$ of its support}. Since $P_0$ is a well behaved function,
namely it varies smoothly around its mean, we expect then that typical distributions
in $\cQp$ with have a {lot of} modes. 

We say that a distribution $Q\in\cQp$ has a mode at $i$ if 
\[
Q(i-1)<Q(i)>Q(i+1) \text{ or }Q(i-1)>Q(i)<Q(i+1).
\]

We consider $i\in\cI_l\ed[n/2-4\sqrt n, n/2+4\sqrt n]$. Then by 
Lemma~\ref{lem:chernoff_indicator} for $P_0$, 
\[
P_0(\cI_l)\ge 1-2e^{-8}>0.99.
\]

Note that for $i\in\cI_l$, 
\begin{align}
\label{eqn:ratio}
\frac{P_0(i+1)}{P_0(i)} =
\frac{\absz-i}{i+1}\in\left[1-\frac{18}{\sqrt n}, 1+\frac{18}{\sqrt
    n}\right]. 
\end{align}

We need the following simple result to show that most distributions
in $\cQp$ are $\dst-$far from all unimodal distributions. 
\begin{Claim}
\label{clm:updown}
Suppose $a_1\ge a_2$ and $b_1\le b_2$, then 
\[
|a_1-b_1|+|a_2-b_2|\ge |a_1-a_2|
\]
\end{Claim}
\begin{proof}
By the triangle inequality, 
\[
|a_1-b_1|+|b_2-a_2|\ge |a_1-a_2+b_2-b_1|\ge a_1-a_2. \notag
\]
Using $a_1\ge a_2$ proves the result. 
\end{proof}

Consider any unimodal distribution $R$ over $[n]$. Suppose its unique mode
is at $j$. Suppose $j \costasnote{\ge} n/2$ \costasnote{(the other possibility is treated symmetrically)} and $R$ is increasing until $j$.  Then for
$Q\in\cQp$, let $\costasnote{\cI_l \ni i}<j$ be such that $Q(i)=P_0(i)\cdot(1+c\cdot\dst)$ and
$Q(i+1)=P_0(i+1)\cdot(1-c\cdot\dst)$. If $c>200$ and
$\dst>100/\sqrt n$, then by~\eqref{eqn:ratio}, $Q(i+1)<Q(i)$ \costasnote{(for large enough $n$)}, and
therefore
\begin{align*}
&|Q(i+1)-R(i+1)|+|Q(i)-R(i)|\\\ge&Q(i)-Q(i+1)\\
=&P_0(i)\cdot(1+c\cdot\dst)-P_0(i+1)\cdot(1-c\cdot\dst)\\
\ge& P_0(i)\cdot c\dst.
\end{align*}
This can be used to lower bound the $\ell_1$ distance from \costasnote{a typical distribution $Q \in \cQp$ to any
unimodal distribution}. Simple Chernoff bounds show that a
randomly chosen string of length $\Theta(\sqrt n)$ over $\{+1, -1\}$
has $\Theta(\sqrt n)$ occurrences of +1-1 and -1+1 in consecutive
locations with high probability. Moreover, note that in the interval
$\cI_l$, $P_0(i)=\Theta(1/\sqrt \absz)$.  Using this along with the bound above
shows that taking $c$ large enough proves that a random distribution
in $\cQp$ is $\dst-$far from all unimodal distributions with high probability.

{\bf Proof of Item (c):}
We consider the distribution obtained
by picking a distribution uniformly from $\cQp$ and generating 
$K=\poi \nsmp$ samples from it. (By the concentration of
the Poisson distribution, it suffices to prove a lower bound w.r.t. the mean $k$ of the Poisson.) Let
$\bar Q^{k}$ denote the distribution over $\poi \nsmp$ length 
samples thus generated. Since a distribution is chosen
at random, the $z_i$'s are independent of each other. Therefore, 
$K_i$, the number of occurrences of symbol $i$ is independent of all
$K_j$'s except $j=\absz-i$. 
Using this we get the following decomposition
\begin{align*}
&\bar Q^k(\costasnote{K_0=k_0}, \ldots, K_n=k_\absz) \\
=&\prod_{\costasnote{i=0}}^{\absz/2}\bar{Q}^k(K_i=k_i, K_{n-i}=k_{\absz-i}).
\end{align*}
Now $K_i$ and $K_{\absz-i}$ are generated either by \costasnote{$\poi{\lambda_i^-}$, where $\lambda_i^-\ed {\nsmp
  (1-c\dst)P_0(i)}$, or by  $\poi{\lambda_i^+}$, where $\lambda_i^+\ed {\nsmp (1+c\dst)P_0(i)}$ with equal
probability}. Therefore:
\begin{align}
\label{eqn:q}
&\bar{Q}^k(K_i=k_i, K_{n-i}=k_{\absz-i})\nonumber\\
=&\frac12[\poid{\lambda_i^+}{k_i}\poid{\lambda_i^-}{k_{\absz-i}}+\poid{\lambda_i^-}{k_i}\poid{\lambda_i^+}{k_{\absz-i}}]\nonumber\\
=& \frac12\frac{e^{-2\nsmp P_0(i)}}{k_i!k_{\absz-i}!}(\nsmp
P_0(i))^{k_i+k_{\absz-i}}\cdot\\
&\left[(1+c\dst)^{k_i}(1-c\dst)^{k_{\absz-i}}+(1-c\dst)^{k_i}(1+c\dst)^{k_{\absz-i}})\right]\nonumber.
\end{align}
Let $P_0^k$ denote distribution over $\poi{k}$ samples from 
the Binomial $P_0$. 
By independence of multiplicities, 
\begin{align*}
&P_0^k (K_1=k_1, \ldots, K_n=k_\absz)\\
 =& \prod_{i=1}^{\absz} \poid{\nsmp P_0(i)}{k_i}\\
=&\frac{e^{-\nsmp P_0(i)}}{k_i!}(\nsmp P_0(i))^{k_i}.
\end{align*}
Our objective is to bound $\ellone{P_0^k}{\bar Q^k}$. We use the  
following. 
\begin{Lemma}[\cite{DevroyeL01}]
For any distributions $P$ and $Q$
\[
2\ellone{P}{Q}^2\le \log \EE_Q\left[\frac{Q}{P}\right].
\]
\end{Lemma}
\begin{proof}
By Pinsker's Inequality~\cite{CoverT06}, and concavity of logarithms, 
\[
2\ellone{P}{Q}^2\le KL(Q,P) = \EE_Q\left[\log \frac{Q}{P}\right]\le \log
\left[\EE_Q\frac{Q}{P}\right].
\]
\end{proof}

We consider the ratio of $\bar Q^k$ to $P_0^k$, and obtain
\begin{align*}
&\frac{\bar Q^k(K_0=k_0, \ldots, K_n=k_\absz)}{P_0^k(K_0=k_0, \ldots,
  K_n=k_\absz)}\\
=&\prod_{\costasnote{i=0}}^{\costasnote{{\absz \over 2}-1}}\frac{\bar{Q}^k(K_i=k_i,
  K_{n-i}=k_{\absz-i})}{P_0^k(K_i=k_i)P_0^k(K_{n-i}=k_{\absz-i})}\\
= &\costasnote{\prod_{{i=0}}^{{{\absz \over 2}-1}}}{(1+c\dst)^{k_i}(1-c\dst)^{k_{\absz-i}}+(1-c\dst)^{k_i}(1+c\dst)^{k_{\absz-i}} \over 2}
\end{align*}
where we used~\eqref{eqn:q}. We can use this now to calculate the following expectation
\begin{align*}
&\EE_{\bar Q^k}\left[\frac{\bar Q^k}{P_0^k}\right]\\
=& 
\costasnote{\prod_{i=0}^{\absz/2-1}}\left[\sum_{k_i\ge0, k_{\absz-i}\ge0}\bar{Q}^k(K_i=k_i,
  K_{n-i}=k_{\absz-i})\cdot\right.\\ &~~~\left.\frac 12\left((1 \!+\! c\dst)^{k_i}(1 \!-\! c\dst)^{k_{\absz-i}}\!\!+\! (1 \!-\! c\dst)^{k_i}(1 \!+\! c\dst)^{k_{\absz-i}}\right)\right].
\end{align*}
For $X\sim \poi \lambda$, elementary calculus shows that 
\[
\EE[a^X]=e^{\lambda(a-1)}.
\]
Combining with~\eqref{eqn:q}, and using
$P_0(i)=P_0(\absz-i)$, the above expression simplifies to 
\begin{align*}
\EE_{{\bar Q}^k}&\left[\frac{{\bar Q}^k}{P_0^k}\right]\\
~~=&\costasnote{\prod_{i=0}^{\absz/2-1}} \frac12\left[e^{c\dst \nsmp (1+c\dst) P_0(i)}e^{-c\dst \nsmp
    (1-c\dst)P_0(i)}\right.\\&\qquad \quad \left.+e^{-c\dst \nsmp (1+c\dst) P_0(i)}e^{c\dst
    \nsmp (1-c\dst) P_0(i)}\right]\\
~~=& \costasnote{\prod_{i=0}^{\absz/2-1}}\frac12\left[e^{2c^2\dst^2\nsmp P_0(i)}+e^{-2c^2\dst^2\nsmp
    P_0(i)}\right]\\
\le& e^{2c^4\dst^4 \nsmp^2\sum_{i=0}^{\costasnote{\absz/2-1}}P_0(i)^2},
\end{align*}
where the last step uses, $e^x+e^{-x}\le 2e^{x^2/2}$.
Using Stirling's approximation, we get:
\begin{align*}
\sum_{i=0}^{\costasnote{\absz/2-1}}P_0(i)^2 \le\max_{i} P_0(i)\le P_0\left(\frac\absz2\right)={n
  \choose \frac \absz2}\frac1{2^n}\le \frac
1{\sqrt{n}}.
\end{align*}
Therefore, 
\begin{align*}
\ellone{P_0^k}{\bar Q^k}^{\costasnote{2}}\le \frac{c^4\dst^4 \nsmp^2}{\sqrt{\absz}}.
\end{align*}

Unless $\nsmp=\Omega(\absz^{1/4}/\dst^2)$, there is no test to 
distinguish a distribution picked uniformly from $\cQp$
versus $P_0$. This proves item $(c)$.

\section*{Acknowledgements}
The authors thank Piotr Indyk, Gautam Kamath,  
and Ronitt Rubinfeld for helpful discussions and thoughts. Jayadev
Acharya thanks Ananda Theertha Suresh for valuable suggestions. 

{\footnotesize{\bibliographystyle{alpha}
\bibliography{masterref}}}
\newpage
\end{document}